\documentclass[english,11pt,reqno]{smfart}
\usepackage{color}
\usepackage{amsthm}
\usepackage{amsmath}
\usepackage{amsfonts}
\usepackage{amstext}
\usepackage[latin1]{inputenc}
\usepackage{amscd}
\usepackage{latexsym}
\usepackage{bm}
\usepackage{amssymb}
\usepackage[all,cmtip]{xy}
\usepackage{euscript}
\usepackage{a4wide}
\usepackage{wasysym}
\usepackage{graphicx}

\def\di{\displaystyle}

\newtheorem{theorem}{Theorem}
\newtheorem{definition}{Definition}
\newtheorem{lemma}{Lemma}
\newtheorem{proposition}{Proposition}

\newtheorem{remark}{Remark}

\newcommand{\N}{\mathbb{N}}
\newcommand{\Z}{\mathbb{Z}}
\newcommand{\R}{\mathbb{R}}

\newcommand{\fonctionsansdef}[3]{\begin{array}[t]{lrcl}#1 :&#2 &\longrightarrow &#3 \end{array}}

\newcommand{\TT}{\mathbb{T}}
\newcommand{\tTT}{\bar{\mathbb{T}}}
\def\TK{\mathbb{T}^\kappa}
\def\Tk{\mathbb{T}_\kappa}
\def\TKk{\mathbb{T}^\kappa_\kappa}

\def\RS{\mathrm{RS}}
\def\LS{\mathrm{LS}}
\def\RD{\mathrm{RD}}
\def\LD{\mathrm{LD}}

\def\Crd{C^0_{\mathrm{rd}}}
\def\Cdrd{C^{1,\Delta}_{\mathrm{rd}}}

\begin{document}
\setcounter{tocdepth}{3}
\title{A Time scales Noether's theorem}
\author{Baptiste Anerot$^2$, Jacky Cresson$^{1,2}$ and Fr\'ed\'eric Pierret$^3$}

\maketitle

\begin{abstract}
We prove a time scales version of the Noether's theorem relating group of symmetries and conservation laws. Our result extends the continuous version of the Noether's theorem as well as the discrete one and corrects a previous statement of Bartosiewicz and Torres in \cite{BT}.
\end{abstract}

\noindent {\tiny $^1$ SYRTE, Observatoire de Paris, CNRS UMR 8630, 77 Avenue Denfert-Rochereau, 75014 Paris, France}

\noindent {\tiny $^2$ Laboratoire de Math\'ematiques Appliqu\'ees de Pau, Universit\'e de Pau et des Pays de l'Adour,}

\noindent {\tiny  avenue de l'Universit\'e, BP 1155, 64013 Pau Cedex, France}

\noindent {\tiny $^3$ IMCCE, Observatoire de Paris, CNRS UMR 8028, 77 Avenue Denfert-Rochereau, 75014 Paris, France}

\tableofcontents

\section{Introduction}

The calculus on time-scales initiated by Stefan Hilger in \cite{hil} gives a convenient way to deal with discrete, continuous or mixed processes using a unique formalism. In 2004, this theory was used by M. Bohner \cite{bohn} and R. Hilscher and V. Zeidan \cite{HZ} to develop a {\it calculus of variations on time scales}. In this context, many natural problems arise. One of them is to generalize to the time scales setting classical results of the calculus of variation in the continuous case. One of these problem is to obtain a time scales analogue of the Noether's Theorem relating group of symmetries and conservation laws.\\

The aim of this article is precisely to derive a time scales version of the Noether's theorem. We refer to the books of Olver \cite{olver} and Jost \cite{jost} for the classical case.  This problem was initially considered by Z. Bartosiewicz and D.F.M. Torres in \cite{BT} but both the result and the proof are incomplete. In the following, we follow the strategy of proof proposed in \cite{BT} consisting in deriving the Noether's theorem for transformations depending on time from the easier result obtained for transformations without changing the time. In \cite{ca}, we call {\it Jost's method} this way of proving the Noether's theorem as a classical reference is contained in the book \cite{jost}. 

\subsection{Main result}

Our main result can be formulated as follows. \\

Let $\TT$ be a bounded time scale with $a = \min (\TT)$, $b = \max (\TT)$ and $\mathrm{card} (\TT) \geq 3$. We denote by $\rho$ and $\sigma$ the backward and forward jump operator (see Definition \ref{jump}). We set $\TK = \TT \backslash ]\rho(b),b]$, $\Tk = \TT \backslash [a,\sigma(a)[$ and $\TKk = \TK \cap \Tk$. We denote by $C^{1,\Delta}_{\mathrm{rd}}(\TT)$ the set of $\Delta$-differentiable functions on $\TK$ with rd-continuous $\Delta$-derivative (see Definition \ref{functio}).\\

Let us consider a functional $\mathcal{L} :C^{1,\Delta}_{\mathrm{rd}}(\TT) \rightarrow \R$ defined by 
$$\mathcal{L} (q) = \di \int_a^b L(t,q(t),\Delta q(t))\Delta t ,$$
where \fonctionsansdef{L}{[a,b]\times \R^d \times\R^d}{\R} is a Lagrangian. The critical point of $\mathcal{L}$ are solutions of the time-scale Euler-Lagrange equation (see \cite{bourdin1}):
\begin{equation}
\label{tsel}
\nabla \left[ \dfrac{\partial L}{\partial v} (t,q(t),\Delta q(t)) \right] =  \nabla \sigma (t) \dfrac{\partial L}{\partial x} (t,q(t),\Delta q (t)) ,
\end{equation}
for every $t \in \TKk$.\\

Following \cite{BT}, a time-scale Lagrangian functional $\mathcal{L}$ is said to be {\it invariant} under the one-parameter family group $G=\left \{ g_s \right \}_{s\in \R}$ of transformations $g_s (x,t)= (g_s^0 (t) ,g_s^1 (x))$ if and only if for any subinterval $[t_a ,t_b ] \subset [a,b]$ with $t_a, t_b \in \TT$, for any $s \in \R$ and $x\in C^{1,\Delta}_{rd} (\TT )$ one has
\begin{equation}
\label{invariance}
\int_{t_a}^{t_b}L\left(t,x(t),\Delta  x(t)\right)\Delta t = \int_{\tau_a}^{\tau_b}L\left(\tau,g^1_s\circ x \circ(g_s^0)^{-1}(\tau),\Delta_{\tTT} \left(g^1_s\circ x \circ(g_s^0)^{-1}(\tau)\right)\right)\Delta_{\tTT}\tau
\end{equation}
where $\tau_a=g^0_s(t_a)$ and $\tau_b=g^0_s(t_b)$.\\

In the following, we need the notion of {\it admissible} group of symmetries which corresponds to one-parameter group of diffeomorphisms satisfying:

\begin{itemize}
	\item[$\bullet$] the set defined by $\di \tTT_s=g^0_s(\TT)$ is a time-scale for all $s\in\R$,
	\item[$\bullet$] the function $g_s^0$ is strictly increasing,
	\item[$\bullet$] $\di \Delta_{\tTT_s}\left(g_s^0\right)^{-1}$ exist,
	\item[$\bullet$] $\di \Delta g_s^0 \neq 0$ and $\Delta g_s^0$ is rd-continuous.
\end{itemize}

Our main result is the following version of the time-scale Noether's theorem:

\begin{theorem}[Time-scale Noether's theorem]
\label{main}
Suppose $G=\{ g_s (t,x)=(g_s^0 (t) ,g_s^1 (x) )\}_{s\in \R}$ is an admissible one parameter group of symmetries of the variational problem $$\di\mathcal{L} (x)=\di\int_a^b L\left(t,x(t),\Delta x(t)\right)\, \Delta t$$ 
and
\begin{equation}
X= \zeta (t) \di\frac{\partial}{\partial t} +\xi (x) \di\frac{\partial}{\partial x} ,
\end{equation}
be the infinitesimal generator of $G$. Then, the function
\begin{equation}
\label{conslaw}
I(t,x)=
\zeta^{\sigma} 
\cdot 
\left [ 
L(\star ) -\partial_v L (\star ) \cdot \Delta x 
\right ] 
+
\xi^{\sigma} 
\cdot 
\partial_v L (\star ) 
+
\di \int_a^t
\zeta 
\left [  
\nabla \sigma  \partial_t L (\star) -\nabla \left ( 
L-\partial_v L \cdot \Delta x \right ) 
\right ] 
\, \nabla t 
,
\end{equation}
is a constant of motion over the solution of the time-scale Euler-Lagrange equation (\ref{tsel}), i.e. that 
\begin{equation}
\nabla \left [ I(t,x(t)) \right ] =0 ,
\end{equation}
for all solutions $x$ of the time-scale Euler-Lagrange equations and any $t\in\Tk$.
\end{theorem}

The proof is given in Section \ref{proof}.\\

In the continuous case $\TT =\R$, one obtain the classical form of the integral of motion
\begin{equation}
I(t,x) = \zeta \left ( L(\star )-\partial_v L (\star) \dot{x} \right ) 
+\xi \partial_v L (\star ) ,
\end{equation}
because the last integral term is reduced to zero. Indeed, on the solutions of the Euler-Lagrange equation one has the identity $\di\partial_t L (\star ) = \di\frac{d}{dt} 
\left ( L(\star) -\partial_v L (\star ) \dot{x} \right )$. \\

In the discrete case, $\TT =\Z$ and transformations without changing time, one recovers the classical integral (see \cite{BCG}, Theorem 12 p.885 and also \cite{lub}):

\begin{equation}
I(x)=\xi^{\sigma} 
\cdot 
\partial_v L (\star ) .
\end{equation}

\subsection{Comments on previous results}

\subsubsection{The Bartosiewicz and Torres result} 

In \cite{BT} the authors obtain a time scales version of the Noether theorem in the shifted version of the calculus of variation on time scales. However, their result can be easily extended to the non shifted case following the same paths. It coincides with our result for transformations without changing  time but differs from it in the other cases.\\ 

As illustrated in Section \ref{examples} with an example and numerical simulations, the result in \cite{BT} is not correct. The reason is that in order to follow their scheme of proof (see Section \ref{commentproof}) the solutions of the Euler-Lagrange equations have to satisfy an auxiliary equation given for all $t\in\TKk$ by (see Lemma \ref{key_ts} in Section \ref{proof}):
\begin{equation}
\label{condi}
\quad \nabla \sigma(\tau)\partial_t L  (\star )+\nabla\left(\Delta x(t) \partial_v L (\star ) -L(\star ) \right)=0,
\end{equation}
which is precisely the quantity under the $\nabla$-antiderivative. This quantity is discussed in the next Section.

\subsubsection{The second Euler-Lagrange equation approach}
 
As already noted, in the continuous case $\TT =\R$, condition (\ref{condi}) is well known and corresponds to the {\it second Euler-Lagrange equation} or the {\it Dubois-Raymond necessary optimality condition}. A time-scales analogue of the second Euler-Lagrange equation was derived by Bartosiewicz, Martins and Torres in (\cite{BMT}, Theorem 5 p.12) leading to another proof of the time-scale Noether's theorem (see \cite{BMT},Section 4, Theorem 6).\\
 
As already said, the result in \cite{BT} is wrong without additional assumptions. As a consequence, we believe that the time scales second Euler-Lagrange equation in \cite{BMT} must be taken with care.

\subsection{A time scales Jost's method of proof}
\label{commentproof}

The approach used by Bartosiewicz and Torres to prove their time scales Noether's theorem is an adaptation of a method which can be found in the classical Textbook by J. Jost and X. Li-Jost \cite{jost} on the calculus of variations. Formally, the idea is very simple. One introduce an extra variable corresponding to the time variable in order to transform the case of invariance under transformations changing time to a case of invariance without changing the "time" variable for an extended Lagrangian which is explicitly constructed from the initial Lagrangian. The corresponding Noether's theorem then follows from the one for transformations without changing time which is easier. We refer to \cite{jost} for more details.\\

However, as in the fractional case\footnote{This work was in fact suggested by a recent article \cite{FM} showing that the fractional Noether theorem proved by Frederico and Torres in \cite{FT} is wrong. However, the article \cite{FM} does not provide a clear understanding of where and why the result is not correct. The second author and A. Szafranska have analysed in \cite{ca} the proof given in \cite{FT} which is an adaptation of the Jost's method to the fractional calculus of variations. Several problems was then pointed out which can occur when generalizing the Jost's method to another framework.}, where the same method of proof were used, several problems arise when adapting the method of Jost to the time-scale case. In particular, one must be very careful with the validity of the change of variables and the fact that one can used the time-scale Noether's theorem for transformations without changing time. In particular, the proof proposed in \cite{BT} does not work precisely because one can not use the autonomous version of the Noether's theorem but only the infinitesimal invariance characterization (see Section \ref{proof}, Lemma \ref{key_ts} and after).\\

It must be pointed out that there exists several way to prove the Noether's theorem. However, we decide to follow the same strategy of Bartosiewicz and Torres in \cite{BT} because this method is very elegant and many other generalizations are based on it. As a consequence, the problems that we are discussing will be of importance for other works.

\subsection{Plan of the paper}

The plan of the paper is as follows. In Section \ref{remind}, we remind some definitions and notations about time-scales and give some particular statements about the chain rule formula and the substitution formula for $\Delta$-derivative in the time-scales setting. Section \ref{proof} gives the proof of our main result. The proof of several technical Lemmas are given in Section \ref{technical}. In Section \ref{examples}, we discuss an example first studied by Bartosiewicz and Torres in \cite{BT}. We compare the quantity that we have obtained with the one derived in \cite{BT} using a numerical integration. In particular, it shows that the conservation law obtained in \cite{BT} does not give an integral of motion contrary to the quantity obtained using our Theorem. 

\section{Preliminaries on time scales}
\label{remind}

In this Section, we remind some results about the chain rule formula, the change of variable formula for $\Delta$-antiderivative which will be used during the proof of the main result. We refer to \cite{agar2,bohn,bohn3,bourdin2} and references therein for more details on time scale calculus. \\

\begin{definition}
\label{jump}
The backward and forward jump operators $\rho, \sigma : \TT \longrightarrow \TT$ are respectively defined by:
	\begin{equation*}
	\forall t \in \TT, \; \rho (t) = \sup \{ s \in \TT, \; s < t \} \; \text{and} \; \sigma (t) = \inf \{ s \in \TT, \; s > t \},
	\end{equation*}
	where we put $\sup \emptyset = a$ and $\inf \emptyset = b$.
\end{definition}

\begin{definition}
	A point $t \in \TT$ is said to be left-dense (resp. left-scattered, right-dense and right-scattered) if $\rho (t) = t$ (resp. $\rho (t) < t$, $\sigma (t) = t$ and $\sigma (t) > t$).
\end{definition}

Let $\LD$ (resp. $\LS$, $\RD$ and $\RS$) denote the set of all left-dense (resp. left-scattered, right-dense and right-scattered) points of $\TT$.

\begin{definition}
	The graininess (resp. backward graininess) function $\fonctionsansdef{\mu}{\TT}{\R^+}$ (resp. $\fonctionsansdef{\nu}{\TT}{\R^+}$) is defined by $\mu(t) = \sigma (t) -t$ (resp. $\nu(t) = t- \rho (t)$) for any $t \in \TT$.
\end{definition}

Let us recall the usual definitions of $\Delta$- and $\nabla$-differentiability.

\begin{definition}
	A function $\fonctionsansdef{u}{\TT}{\R^n}$, where $n \in \N$, is said to be $\Delta$-differentiable at $t \in \TK$ (resp. $\nabla$-differentiable at $t \in \Tk$) if the following limit exists in $\R^n$:
	\begin{equation}
	\lim\limits_{\substack{s \to t \\ s \neq \sigma (t) }} \dfrac{u(\sigma(t))-u(s)}{\sigma(t) -s} \; \left( \text{resp.} \; \lim\limits_{\substack{s \to t \\ s \neq \rho (t) }} \dfrac{u(s)-u(\rho (t))}{s-\rho(t)} \right).
	\end{equation}
	In such a case, this limit is denoted by $\Delta u (t)$ (resp. $\nabla u (t)$).
\end{definition}

\begin{proposition}
	\label{rappeldelta2}
	Let $\fonctionsansdef{u}{\TT}{\R^n}$. Then, $u$ is \\
	$\Delta$-differentiable on $\TK$ with $\Delta u = 0$ if and only if there exists $c \in \R^n$ such that $u(t) =c$ for every $t \in \TT$.
\end{proposition}

The analogous results for $\nabla$-differentiability are also valid.

\begin{definition}
\label{functio}
A function $u$ is said to be rd-continuous (resp. ld-continuous) on $\TT$ if it is continuous at every $t \in \RD$ (resp. $t \in \LD$) and if it admits a left-sided (resp. righ-sided) limit at every $t \in \LD$ (resp. $t \in \RD$).
\end{definition}

We respectively denote by $\Crd(\TT)$ and $\Cdrd(\TT)$ the functional spaces of rd-continuous functions on $\TT$ and of $\Delta$-differentiable functions on $\TK$ with rd-continuous $\Delta$-derivative. \\

Let us denote by $\int \Delta \tau$ the Cauchy $\Delta$-integral defined in \cite[p.26]{bohn} with the following result (see {\cite[Theorem 1.74 p.27]{bohn}}):

\begin{theorem}
	For every $u \in \Crd(\TK)$, there exist a unique $\Delta$-antiderivative $U$ of $u$ in sense of $\Delta U = u$ on $\TK$ vanishing at $t=a$. In this case the $\Delta$-integral is defined by
	\begin{equation*}
	U(t) = \int_a^t u(\tau) \Delta \tau
	\end{equation*}
	for every $t \in \TT$.
	\label{thm_antiderivative}
\end{theorem}

We have a time-scale chain rule formula (see \cite[Theorem 1.93]{bohn}).

\begin{theorem}[Time-scale Chain Rule]
\label{tscr}
Assume that \fonctionsansdef{v}{\TT}{\R} is strictly increasing and $\tilde{\TT}:=v(\TT)$ is a time-scale. Let \fonctionsansdef{w}{\tilde{\TT}}{\R}. If $\Delta v(t)$ and $\Delta_{\tilde{\TT}}(v(t))$ exist for $t\in\TK$, then
\begin{equation}
\Delta\left(w\circ v\right) = \left(\Delta_{\tilde{\TT}}\circ v\right) \Delta v
\end{equation}
\end{theorem}
With the time-scale chain rule, we obtain a formula for the derivative of the inverse function (see \cite[Theorem 1.97]{bohn}).
\begin{theorem}[Derivative of the inverse]
Assume that \fonctionsansdef{v}{\TT}{\R} is strictly increasing and $\tilde{\TT}:=v(\TT)$ is a time-scale. Then
\begin{equation}
\frac{1}{\Delta v}=\Delta_{\tilde{\TT}}\left(v^{-1}\right)\circ v
\end{equation}
at points where $\Delta v$ is different from zero.
\end{theorem}
Another formula from the chain rule is the substitution rule for integrals (see \cite[Theorem 1.98]{bohn}).
\begin{theorem}[Substitution]
Assume that \fonctionsansdef{v}{\TT}{\R} is strictly increasing and $\tilde{\TT}:=v(\TT)$ is a time-scale. If \fonctionsansdef{f}{\TT}{\R} is a rd-continuous function and $v$ is differentiable with rd-continuous derivative, then for $a,b\in\TT$,
\begin{equation}
\int_{a}^{b} f(t)\Delta v(t)\Delta t = \int_{v(a)}^{v(b)}\left(f\circ v^{-1}\right)(s)\Delta_{\tilde{\TT}} s.
\end{equation}
\end{theorem}

\section{Proof of the main result}
\label{proof}

We first rewrite the invariance relation (\ref{invariance}) in order to have the same domain of integration.

\begin{lemma}
\label{changement_bornes_ts}
Let $\mathcal{L}$ be a time-scale Lagrangian functional invariant under the action of the group of diffeomorphisms $g$. Then, we have
\begin{equation}
\label{invar}
\int_{a}^{b}L\left(t,x(t),\Delta  x(t)\right)\Delta  t = \int_{a}^{b}L\left (g^0_s(t),(g^1_s\circ x)(t),\Delta  \left(g_s^1\circ x \right)(t) \frac{1}{\Delta  g_s^0(t)}\right )\Delta  g_s^0(t)\Delta  t.
\end{equation}
\end{lemma}

The proof is given in Section \ref{proof_changement_bornes_ts}. \\

As for the classical case, we construct an extended Lagrangian functional $\bar{\mathcal{L}}$ associated with the autonomous Lagrangian $\bar{L}$ as follows: \\

Let $\bar{\mathcal{L}}: C^{2}_{\Delta, \nabla}([a,b],\R) \times C^{2}_{\Delta, \nabla}([a,b],\R) \rightarrow \R$ defined by
\begin{equation}
\bar{\mathcal{L}}(t,x) =\int_{a}^{b} \bar{L}\left(t(\tau),x(t(\tau)),\Delta_{\tTT}t(\tau),\Delta_{\tTT}x(t(\tau))\right)\Delta \tau.
\end{equation}
where $\tilde{L}: \R\times \R^{d}\times\R\times \R^{d} \rightarrow \R$ is defined by
\begin{equation}
\bar{L}(t,x,w,v)=L\left(t,x,\frac{v}{w}\right)w.
\end{equation}
which is the same as the classical case. We define the \emph{time-scale bundle path class} denoted by $\bar{\mathsf{F}}$ and defined by
\begin{equation}
\bar{\mathsf{F}} = \{(t,x)\in C^{2}_{\Delta, \nabla}([a,b],\R) \times C^{2}_{\Delta, \nabla}([a,b],\R) \ ; \ \di \tau \longmapsto(t(\tau),x(\tau))=(\tau,x(\tau)\}.
\end{equation}

We have the following proposition:

\begin{proposition}
The restriction of the Lagrangian function $\bar{\mathcal{L}}$ to a path $\gamma=(t,x) \in \bar{\mathsf{F}}$ satisfies
\begin{equation}
\label{restri_equal_ts}
\bar{\mathcal{L}}(t,x)=\mathcal{L}(x).
\end{equation}
\end{proposition}

\begin{proof}
Let $\gamma=(t,x) \in \bar{\mathsf{F}}$. By definition, we have
\begin{equation*}
\tilde{L}\left(t(\tau),x(\tau),\Delta_{\tTT}t(\tau),\Delta_{\tTT}x(t(\tau)\right)=L\left(t(\tau),x(t(\tau)),\Delta_{\tTT}x(t(\tau))\frac{1}{\Delta_{\tTT}t(\tau)}\right)\Delta_{\tTT}t(\tau).
\end{equation*}
As $\gamma$ is a bundle path, we have $t(\tau)=\tau$ and $\Delta_{\tTT}t(\tau)=1$. In consequence, $\tilde{\TT}=\TT$ and we obtain
\begin{equation*}
\tilde{\mathcal{L}}(t,x)=\int_{a}^{b} \tilde{L}\left(t(\tau),x(t(\tau)),\Delta_{\tTT}t(\tau),\Delta_{\tTT}x(t(\tau))\right)\Delta_{\tTT} \tau=\int_{a}^{b} L\left(\tau,x(\tau),\Delta x(\tau)\right)\Delta\tau = \mathcal{L}(x).
\end{equation*}
\end{proof}

In order to formulate the time-scale Euler-Lagrange equation for the extended autonomous Lagrangian, we need to have the $\nabla_{\tTT_s}$-differentiability of $\bar{\sigma}$. We have:

\begin{lemma}
	\label{nabla_diff_TTS}
	Let $s\in\R$. Let $\bar{\sigma}_s$ to be the forward jump operator over $\tTT_s$. Assume that $\sigma$ is $\nabla$-differentiable on $\Tk$ then $\bar{\sigma}_s$ is $\nabla_{\tTT_s}$-differentiable on $(\tTT_s)^\kappa_\kappa$.
\end{lemma}

\begin{proof}
	Let $s\in \R$. By definition, $\di \bar{\sigma} = \sigma\circ g^0_s$ and we have $\sigma\circ g^0_s= g^0_s \circ \sigma$. As $g_s^0$ is $\Delta$-differentiable on $\TK$ and $\sigma$ is $\nabla$-differentiable on $\Tk$ then, from Theorem \ref{tscr}, we obtain that $g^0_s \circ \sigma$ is $\nabla$-differentiable on $\TKk$. As $\tTT_s=g^0_s(\TT)$, we obtain the result.
\end{proof}

In what follows, we assume that $\sigma$ is $\nabla$-differentiable on $\Tk$.

\begin{lemma}
\label{key_ts}
A path $\gamma=(t,x)\in\bar{\mathsf{F}}$ is a critical point of $\bar{\mathcal{L}}$ if, and only if, $x$ is a critical point of $\mathcal{L}$ and for all $t\in\TKk$ we have
	\begin{equation}
	(\boldsymbol{\hexstar}) \quad \nabla \sigma(\tau)\frac{\partial L}{\partial t}(t,x(t),\Delta x(t))+\nabla\left(\Delta x(t)\frac{\partial L}{\partial v}(t,x(t),\Delta x(t)) -L(t,x(t),\Delta x(t))\right)=0.
	\end{equation}
\end{lemma}

The proof is given in Section \ref{proof_key_ts}. \\

Contrary to the continuous case, Lemma \ref{key_ts} implies that extended solutions of the initial Lagrangian are not automatically solutions of the extended Euler-Lagrange equation. This implies that one can not use the Noether's theorem but only the infinitesimal invariance criterion.

\begin{lemma}
\label{invariance_Ltilde_ts}
Let $\mathcal{L}$ be a time-scale Lagrangian functional invariant under the one-parameter group of diffeomorphisms $g$. Then, the time-scale Lagrangian functional $\bar{\mathcal{L}}$ is invariant under the one-parameter group of diffeomorphisms $g$ over $\bar{\mathsf{F}}$.
\end{lemma}

The proof is given in Section \ref{proof_invariance_Ltilde_ts}.\\

We deduce from Lemma \ref{invariance_Ltilde_ts} and the {\it necessary condition of invariance} given in (\cite{BT},Theorem 2 p.1223) that 
\begin{equation}
\label{main1}
\partial_t L (\star ) . \zeta +\partial_x L (\star ) . \xi 
+\partial_v L (\star ) .\Delta \xi 
+\left (  L (\star ) -\partial_v L (\star ) .\Delta x \right ) .\Delta \zeta =0.
\end{equation}
Multiplying equation \eqref{main1} by $\nabla \sigma$ and using the Time scales Euler-Lagrange equation \eqref{tsel}, we obtain  
\begin{equation}
\partial_t L (\star ) . \nabla \sigma .  \zeta +\nabla \sigma . \partial_v L (\star ) . \Delta [\xi] +\nabla \left [ \partial_v L (\star ) \right ] . \xi 
+\left (  L (\star ) -\partial_v L (\star ) .\Delta x \right ) .\nabla \sigma . \Delta \zeta =0 .
\end{equation}
Using the relation 
\begin{equation}
\label{leib2}
\nabla \left [ f^{\sigma} g \right ] =
\nabla \sigma \Delta [f] . g +f .\nabla [g] ,
\end{equation}
we have 
\begin{equation}
\partial_t L (\star ) . \nabla \sigma . \zeta +\nabla \left [ \partial_v L (\star ) . \xi^{\sigma} \right ]  
+\left (  L (\star ) -\partial_v L (\star ) .\Delta x \right ) .\nabla \sigma . \Delta \zeta =0 .
\end{equation}
Trying to be as close as possible to the continuous case, we can use again relation (\ref{leib2}) on the last term. We obtain
\begin{equation}
\partial_t L (\star ) . \nabla \sigma .  \zeta +\nabla \left [ \partial_v L (\star ) . \xi^{\sigma} \right ]  
+\nabla \left [  \zeta .\left ( L (\star ) -\partial_v L (\star ) .\Delta x \right ) . \zeta^{\sigma} \right ] -\zeta .\nabla \left [  L (\star ) -\partial_v L (\star ) .\Delta x \right ] =0 .
\end{equation}
Taking the $\nabla$ antiderivative of this expression, we deduce the conservation law (\ref{conslaw}). This concludes the proof.

\section{The Bartosiewicz and Torres example}
\label{examples}

We consider the example of Lagrangian given in \cite{BT} and we illustrate our result with respect to the result given in \cite{BT}. Let $N\in\N^{*}$, $a,b\in \R$ with $a<b$ and let $h=(b-a)/N$. We consider the time-scale $\TT=\{t_k, \ k=0,\cdots N\}$ where $t_k=a+kh$.

\subsection{Invariance and a conservation law}

We consider the Lagrangian introduced in \cite{BT}
\begin{equation}
\label{exemple-delfim}
L(t,x, v)=\frac{x^2}{t} + tv^2
\end{equation}
for $x,v \in\mathbb{R}$. 

\begin{lemma} 
The Lagrangian functional associated to (\ref{exemple-delfim}) is invariant under the family of transformation $G=\{ \phi_s (t,x)=(t e^s , x )\}_{s\in \R}$ where its infinitesimal generator $X$ is given by
\begin{equation}
\zeta(t)=t \quad \mbox{\rm and} \quad \xi(x)=0.
\end{equation}
\end{lemma}

\begin{proof}
Indeed, we have $L\left ( t e^s , x , \di\frac{\Delta x}{e^s} \right ) e^s =\left ( 
\di\frac{x^2}{t e^s} +te^s \di\frac{(\Delta x)^2}{e^{2s}} \right )e^s =L(t,x,\Delta x)$ so that condition \eqref{invar} is satisfied.
\end{proof}

In our case, the (non-shifted) Euler--Lagrange equation associated with $L$ is given by
\begin{equation}
	\nabla \left(t \Delta x(t) \right) = \frac{x}{t},
\end{equation}
and our time-scale Noether's theorem generates the following conservation law
\begin{equation}
	I(t,x,v)=\sigma(t)\left(\frac{x^2}{t}-t v^2\right) + \int_a^t \left[ -\frac{x^2}{t}+tv^2-t\nabla \left(\frac{x^2}{t}-tv^2\right)\right]\nabla t .
\end{equation}
The (shifted) Euler--Lagrange equation associated with $L$ is given by
\begin{equation}
\Delta \left(t \Delta x \right) = \frac{x^\sigma}{t},
\end{equation}
and the time-scale Noether's theorem given in \cite{BT}, generates the following conservation law
\begin{equation}
C(t,x^\sigma,v)=\sigma(t)\left(\frac{(x^\sigma)^2}{t} - t v^2\right).
\end{equation}

\begin{remark}
In \cite{BT}, the authors consider $\mathbb{T}=\{ 2^n : n\in\mathbb{N}\cup\{0\} \}$. In that case, $\sigma(t)=2 t$ for all $t\in \mathbb{T}$, which gives the expression of $C(t,x^\sigma,v)$ in \cite[Example 3]{BT}.
\end{remark}

\subsection{Simulations}

The initial conditions are chosen such that $x(1)=1$ and $\Delta x(1)=0.1$. We display in Figure \ref{result1}, the two quantities computed numerically with $a = 1$,  $b=10$ and $h=10^{-3}$. As we can see, the quantity $I(t,x,\Delta x)$ is a constant of motion over the solution of the time-scale Euler-Lagrange equation. It is clearly not the case for the quantity $C(t,x^\sigma,\Delta x)$ provided by the Noether's theorem in \cite{BT}.

\begin{figure}
\centering
\includegraphics[width=0.7\linewidth]{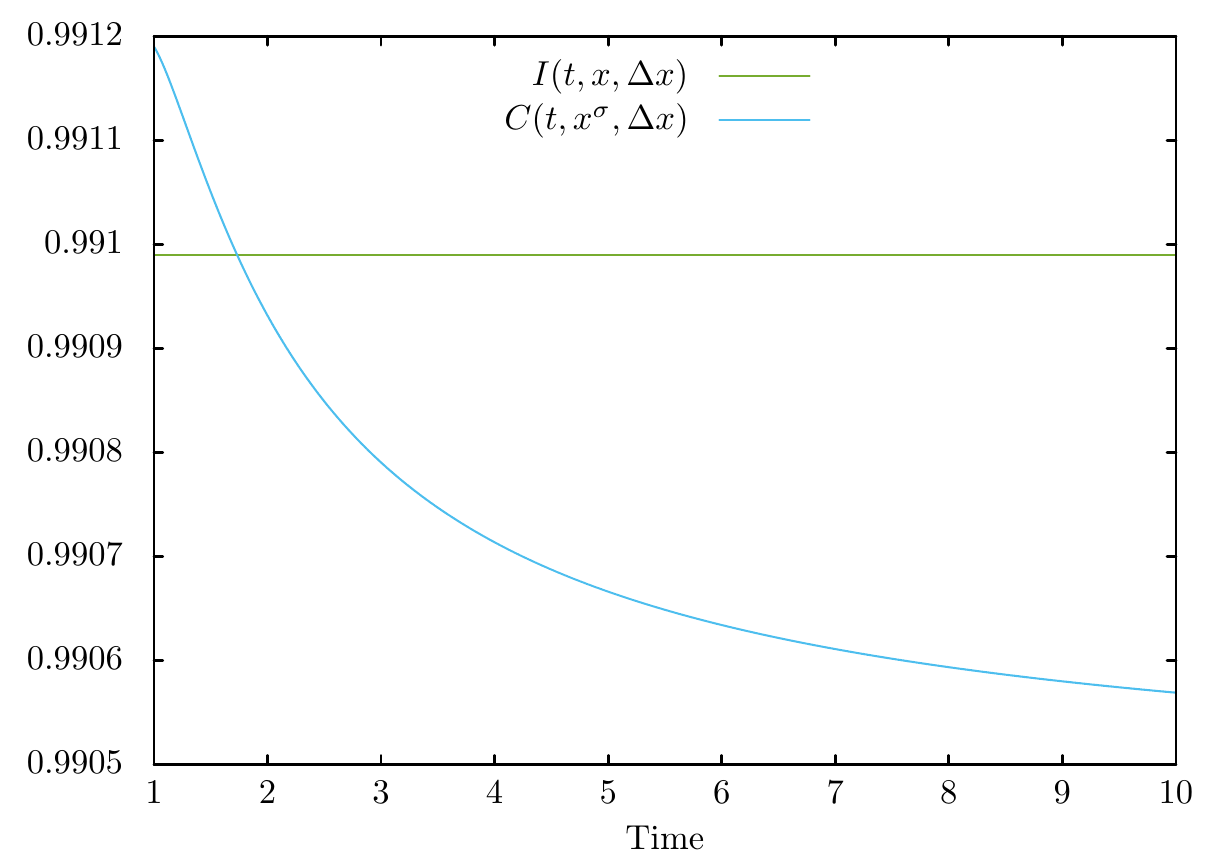}
\caption{}
\label{result1}
\end{figure}

\section{Proof of the technical Lemma}
\label{technical}

\subsection{Proof of Lemma \ref{changement_bornes_ts}}
\label{proof_changement_bornes_ts}

Using the time-scale chain rule, we obtain
\begin{equation*}
\Delta_{\tTT} \left(g^1_s\circ x \circ(g_s^0)^{-1}(\tau)\right)=\Delta \left(g_s^1\circ x \right)(t) \Delta_{\tTT_s}\left(g_s^0\right)^{-1}(\tau).
\end{equation*}
Then, using the time-scale derivative formula for inverse function, we obtain
\begin{equation*}
\Delta_{\tTT} \left(g^1_s\circ x \circ(g_s^0)^{-1}(\tau)\right)=\Delta \left(g_s^1\circ x \right)(t)\frac{1}{\Delta  g_s^0(t)}.
\end{equation*}
Using the change of variable formula for time-scale integrals, we obtain
\begin{align*}
& \int_{\tau_a}^{\tau_b}L\left(\tau,g^1_s\circ x \circ(g_s^0)^{-1}(\tau),\Delta_{\tTT} \left(g^1_s\circ x \circ(g_s^0)^{-1}(\tau)\right)\right)\Delta_{\tTT}\tau \\
&=\int_{a}^{b}L\left (g^0_s(t),(g^1_s\circ x)(t),\Delta  \left(g_s^1\circ x \right)(t) \frac{1}{\Delta  g_s^0(t)}\right )\Delta  g_s^0(t)\Delta  t.
\end{align*}
Finally, using the invariance condition in Equation \eqref{invariance}, we obtain the result. 

\subsection{Proof of Lemma \ref{key_ts}}
\label{proof_key_ts}

For the necessary condition, let $\gamma=(t,x)\in\bar{\mathsf{F}}$ be a critical point of $\bar{\mathcal{L}}$. Then, from Lemma \ref{nabla_diff_TTS} and Equation \eqref{tsel}, it satisfies the following Euler-Lagrange equations
\begin{equation}
(\textrm{EL}^{\nabla \circ \Delta })_{\bar{L}}
\left\{\begin{aligned}
&\nabla_{\tTT}\left[\frac{\partial \bar{L}}{\partial v}(\bar{\star}_\tau)\right]=\nabla \bar{\sigma}(\tau)\frac{\partial \bar{L}}{\partial x}(\bar{\star}_\tau),\\
&\nabla_{\tTT}\left[\frac{\partial \bar{L}}{\partial w}(\bar{\star}_\tau)\right]=\nabla \bar{\sigma}(\tau)\frac{\partial \bar{L}}{\partial t}(\bar{\star}_\tau),
\end{aligned}\right.
\end{equation}
for all $\tau \in (\tTT_s)^\kappa_\kappa$, where $(\bar{\star}_\tau)=\left(t(\tau),x(t(\tau)),\Delta_{\tTT}t(\tau), \Delta_{\tTT}x(t(\tau)\right)$.\\

Let $(\star_\tau)=\left(t(\tau),x(t(\tau)),\Delta_{\tTT}x(t(\tau))\frac{1}{ \Delta_{\tTT} t(\tau)}\right)$. By definition, we have
\begin{align}
&\frac{\partial \bar{L}}{\partial t}(\bar{\star}_\tau)= \frac{\partial L}{\partial t}(\star_\tau) \Delta_{\tTT}t(\tau), &\frac{\partial \bar{L}}{\partial w}(\bar{\star}_\tau) &= L\left(\star_\tau\right) - \Delta_{\tTT}x(t(\tau))\frac{1}{\Delta_{\tTT}t(\tau)} \frac{\partial L}{\partial v}(\bar{\star}_\tau), \label{eq_partialtildets1}\\
&\frac{\partial \bar{L}}{\partial x}(\bar{\star}_\tau) = \frac{\partial L}{\partial x} (\star_\tau)\Delta_{\tTT}t(\tau), &\frac{\partial \bar{L}}{\partial v} (\bar{\star}_\tau)&= \frac{\partial L}{\partial v}(\star_\tau) \label{eq_partialtildets2}.
\end{align}
As $\gamma \in \bar{\mathsf{F}}$, we have $(\star_\tau)=\left(\tau,x(\tau),\Delta x(\tau)\right)$ and $\nabla_{\tTT} \bar{\sigma}(\tau)=\nabla \sigma(\tau)$. In consequence, the first Euler-Lagrange equation is equivalent to
\begin{equation}
\label{EL1_final_ts}
\nabla\left[\frac{\partial L}{\partial v}\left(\star_\tau\right)\right]=\nabla \sigma(\tau)\frac{\partial L}{\partial x}\left(\star_\tau\right).
\end{equation}
for all $\tau \in \TKk$ and the second Euler-Lagrange equation is equivalent to
\begin{equation}
\nabla \sigma(\tau)\frac{\partial L}{\partial t}(\star_\tau)+\nabla\left(\Delta x(\tau) \frac{\partial L}{\partial v} (\star_\tau)-L(\star_\tau)\right)=0,
\end{equation}
for all $\tau \in \TKk$, which corresponds to the condition $(\boldsymbol{\hexstar})$. As Equation \ref{EL1_final_ts} is the Euler-Lagrange equation associated with the Lagrangian functional $\mathcal{L}$, we obtain that $x$ is a critical point of $\mathcal{L}$ and $(\boldsymbol{\hexstar})$ is satisfied. \\

For the sufficient condition, assume that $(\boldsymbol{\hexstar})$ is satisfied and let $x$ be a critical point of $\mathcal{L}$ and let $\gamma$ be the path such that $(t,x) \in \bar{\mathsf{F}}$. Using the same computation as previous, we obtain that $\gamma$ is a critical point of $\tilde{\mathcal{L}}$. This conclude the proof.

\subsection{Proof of Lemma \ref{invariance_Ltilde_ts}}
\label{proof_invariance_Ltilde_ts}
	Let $\gamma=(t,x)\in\bar{\mathsf{F}}$. By definition, we have
	\begin{equation}
	\bar{\mathcal{L}}(g_s(\gamma))=\int_{a}^{b} \bar{L}\left(g^0_s(t(\tau)),g^1_s\circ x (t(\tau)),\Delta_{\tTT_s}g_s^0(t(\tau)),\Delta_{\tTT_s}\left(g^1_s\circ x (t(\tau))\right)\right)\Delta_{\tTT_s}\tau.
	\end{equation}
	Using the definition of $\bar{L}$, the fact that $t(\tau)=\tau$ and $\Delta g_s^0(\tau)\neq0$ for all $\tau \in \TK$, we obtain
	\begin{equation}
	\bar{\mathcal{L}}(g_s(\gamma))=\int_{a}^{b}L\left (g^0_s(\tau),(g^1_s\circ x)(g^0_s(\tau)),\Delta \left(g_s^1\circ x  \right)(\tau) \frac{1}{\Delta g_s^0(\tau)}\right )\Delta g_s^0(\tau)\Delta \tau.
	\end{equation}
	Using the invariance of $\mathcal{L}$ with the Lemma \ref{changement_bornes_ts}, we obtain
	\begin{equation}
	\bar{\mathcal{L}}(g_s(\gamma))=\int_{a}^{b}L\left (\tau,x(\tau),\Delta x(\tau)\right )\Delta \tau.
	\end{equation}
	In consequence, as $\Delta t(\tau)=1$, we obtain
	\begin{align}
	\bar{\mathcal{L}}(g_s(\gamma))=\int_{a}^{b}\bar{L}\left (\tau,x(\tau),1,\Delta x(\tau)\right )d\tau=\bar{\mathcal{L}}(\gamma).
	\end{align}
	This concludes the proof.

\end{document}